\documentclass[12pt, reqno]{amsart}
\usepackage{geometry}
\usepackage{amsthm}
\usepackage{amssymb}
\newtheorem{definition}{Definition}[section]
\newtheorem{theorem}{Theorem}[section]
\newtheorem{lemma}{Lemma}[section]
\usepackage{enumerate}
\usepackage{graphicx}
\graphicspath{{./images/}}
\usepackage[T1]{fontenc}
\usepackage[utf8]{inputenc}
\numberwithin{equation}{section}
\usepackage[nomarkers]{endfloat}
\usepackage{hyperref}

\let\emptyset\varnothing

\title{On The Dynamical Nature of Computation}
\author{Nabarun Mondal}
\address{D.E.Shaw \& Co. India, Hyderabad }
\email{mondal@deshaw.com}
\thanks{Nabarun Mondal : 
Dedicated to my missing geometry teacher Dr. Sushanta Mondal, in memory of DhrubaJyoti Ghosh and 
Dr. Prashanta Kumar Nandi,
my parents : Tapan and Sabita Mondal; 
Big thanks to : Abhishek Chanda and Shweta Bansal : You all have been constant support.
}

\author{Partha P. Ghosh}
\address{Microsoft India, Hyderabad }
\email{parthag@microsoft.com}
\thanks{ Partha. P. Ghosh : 
Dedicated to my parents and family, UJI , without their presence we are nothing.
}


\begin{document}


\begin{abstract}
Dynamical Systems theory generally deals with fixed point iterations
of continuous functions. Computation by Turing machine although is a fixed point iteration 
but is not continuous. This specific category of fixed point iterations can only be 
studied using their orbits. Therefore the standard notion of chaos is not immediately applicable. 
However, when a suitable definition is used, it is found that the notion of chaos and fractal sets exists 
even in computation. It is found that a non terminating Computation will be 
almost surely chaotic, and autonomous learning will almost surely identify fractal only sets.   
\end{abstract}

\maketitle

\begin{section}{Computation As Dynamical System}\label{intro}

In general we associate the term dynamical system to a fixed point iteration 
of a continuous function `$f$' as such:
$$
x_n = f(x_{n-1}).
$$

Turing Machines can be treated as a fixed point iteration with suitable rationalisation of the tape symbols $T$ with $T_n$ denoting the state of the tape at iteration `n':
$$
T_n = C(T_{n-1}) 
$$

Due to rationalisation $\rho(T) \in X_Q$ (definition \ref{rat}) of the tape $T$ to $X_Q$ of  definition \eqref{rat-s} a general computation like fixed point iteration (FPI) is therefore an arbitrary function 
of the following form :
\begin{equation}\label{gen-cc-fpi}
f : X_Q \to  X_Q
\end{equation}

But never the less, this pose a problem to analyse computation as a standard dynamical system, 
which in general is continuous nowhere.

The problem is to find generic dynamical properties of these sort of functions \eqref{gen-cc-fpi}. 
As really nothing is known about the function itself, 
we must then  seek through the characteristic of the possible orbits (definition \ref{orbit}) generated by the class of functions same as computation.

What sort of sequences they would be?
Sequences which can be generated by a function, that is :
$$
\forall x \; x  \to y \text{ and } x  \to z \implies y = z
$$ 
That is, same number (as input) can not produce  (different output),
multiple next numbers in a sequence. This evidently means no repetition of any element is allowed.
We note that these sequences are really all possible sequences those can be generated by any function, even non-computable ones having domain and range in $X_Q$.

\end{section}

\begin{section}{Properties of Sequences in $X_Q$}\label{ns}
We notice that the sequences in $X_Q$ can be either :
\begin{enumerate}
\item{Converging : or Cauchy sequences (CS) (definition \ref{cs}), 
the set is of such sequences would be designated as $\mathcal{C}$.}

\item{ Non Converging : or non Cauchy sequences (NCS), the set is to be called $\mathcal{N}$. } 
\end{enumerate}
 
Every non converging sequence is by definition a mixture of multiple Cauchy subsequences,
which is a trivial thing to demonstrate. Therefore a CS is a special case of NCS having only single 
Cauchy subsequence in it. We shall see that the CS provides certain kind of ``basis'' for the sequences space.

\begin{definition}\label{ss}
\textbf{Sequences Space.}

Given any arbitrary infinite sequence is a mixture of 
$n$ Cauchy sequences $C_i$, 
we can represent the family of such sequences as:
$$
L = (a^1, a^2, ..., a^n)
$$ 
where $a^i$ is the accumulation point of the Cauchy Sequence $C_i$,
with $a^i \le a^j $. 
All sequences mixed upto n different Cauchy Sequences can be represented
in this space.   
Any Cauchy Sequence accumulating into point `$a$' can be represented as :
$$
L_a = (a, a, a,..., a)  
$$  
that is all \emph{co-ordinates} having the same value.
\end{definition}

This definition brings about dimensionality of an arbitrary sequence.
A sequence family having only one Cauchy sequence is one dimensional.
Any finite sequence is zero dimensional.
Any other sequence family is multi dimensional.
Thus, number of accumulation points in a sequence defines the dimensionality.

This geometrical notion brings about a measure depicting hyper-volume
over the sequences space which generalises Lebesgue measure. 

\begin{definition}\label{ms}
\textbf{Measure over Sequences Space.}

Define $V( x_2^k , x_1^k )$ as the hyper volume bounded by surfaces $x^k = x_1^k , x^k = x_2^k$,
with $x_1^k,x_2^k \in \mathbb{R}$.
Suppose $y_k = (x_1^k, x_2^k)$ with $y = ( y_1, y_2, ... y_n )$.
Define $V_n(y)$ as the measure of the set of sequences $y \subseteq \mathbb{R}^n $ in n dimension:
$$
V_n(y) = \bigcap_{k=1}^n  V(x_1^k , x_2^k ,  )
$$ 
\end{definition}

As an example let's restrict the universe of sequences allowing upto the mixing of 2 or less Cauchy sequences. It is obvious that in two dimension ( space where mixing upto two Cauchy Sequences are allowed) the volume is a triangle inside a square and the strict Cauchy Sequences are lying in the diagonal.
Thus the Cauchy sequences gets a zero measure. 
Increasing $n$ does not change the volume of the pure Cauchy Sequences, it remains 0.
This observation brings the next theorem. 
\begin{theorem}\label{cmz}
\textbf{Pure Cauchy Sequences are of Measure Zero.}

The measure of definition \ref{ms} dictates that the pure Cauchy Sequences $\mathcal{C}$ is a null set.
\end{theorem}
\begin{proof}
We notice that measure of the pure Cauchy Sequences $(a,a,a,...,a)$ is the volume of a diagonal 
of an $n$ dimensional hyper cube. A line has a zero volume in every dimension, 
and therefore, irrespective of the dimensionality, 
the set of Cauchy Sequences $\mathcal{C}$ is a measure null set. 

\end{proof}

Now, given that the measure of $\mathcal{C}$ is zero,  given a sequence is picked 
from the set of sequences (which is almost entirely $\mathcal{N}$ ) at random, 
we must have a sequence from $\mathcal{N}$, almost surely.

Therefore, a general observation is general FPI orbit would exhibit mixing of sequences. 

\begin{theorem}\label{fpim}
\textbf{Fixed Point Iterations almost surely would yield Non Cauchy Sequences.}

Almost all orbits resulting from fixed point iterations as of equation \eqref{gen-cc-fpi} 
would be Non Cauchy.  
\end{theorem}
\begin{proof}
See the discussion before.
\end{proof}

\end{section}

\begin{section}{Chaotic Sequences in Computation}\label{chaos-s}

We generally attribute randomness to non determinism.
But, a sequence which is generated by a non computable function, is algorithmically 
random by having undefined Chaitin-Solmonoff-Kolmogorov ( CSK ) complexity \cite{mlpv}.
The difference of this algorithmically random behaviour from \emph{non determinism} 
is that in the generated sequence both $x$ followed by $y$ and $x$ followed by $z$ 
with $y \ne z$ can not be present.
Because, given same input a function can not generate different output. 

Then, we can argue that though it is random, but still it is deterministic, because 
there is a deterministic function, albeit non computable,  whose orbit is that sequence,
and the sequence can not be algorithmically computable.

Even computable sequences, if we do not know the computable function generating it
can be deterministic, and random like : unpredictable. We argue that is what chaos should mean 
in context of computation : loss of long term predictability. 
We must first redefine chaos for the sequences such that our definition 
which is valid for a single sequence  matches up the definition of a chaotic trajectory
in the literature.

After carefully considering the possibilities we put forward the following 
definition for chaotic sequences:

\begin{definition}\label{chaos-s}
\textbf{Chaotic Sequences.}

A sequence $<x_k>$ with $x_k \in X_Q$ generated by a function $f$ is said to be chaotic iff 
it exhibits \textbf{Sensitive dependence on initial condition}:
$$
\forall x_i, x_j \in <x_k> \text { such that } |x_i - x_j| < \epsilon
$$
there exists $n > 0 $ such that:
$$
|x_{i+n} - x_{j+n}| > \epsilon
$$ 
holds, where $x_{k+n} = f^n(x_k)$.
\end{definition}

In essence, the defined property ensures that  
even with full knowledge of the past history of the sequence near x, 
there is no way to predict future behaviour of the sequence.
This is also the very essence of chaos and chaotic orbit.
This is only applicable for the specific orbit, and not the function itself, 
as there is no information about the function.

We now argue that almost surely a sequence picked at random 
from $\mathcal{N}$ would be chaotic.

\begin{theorem}\label{ncc}
\textbf{Almost Surely Non Cauchy Sequences are Chaotic.}

Almost all Non Cauchy sequences are chaotic as definition \ref{chaos-s}.  
\end{theorem}
\begin{proof}

Suppose the sequence is a result of mixing of $n$ Cauchy sequences 
with $C_k$ is accumulating to $a_k$.
Then, the final sequence eventually would be trapped in the neighbourhoods:
$A_k = \{ x : |a_k - x | < \epsilon \} $. We note that the $|a_i - a_j| >> \epsilon $
because if it was not, there is no way to isolate $a_i,a_j$ and hence $A_i,A_j$
should be merged into one, by reducing the number $n$ of Cauchy Sequences.

Now, to demonstrate the sensitive dependence we argue the opposite.
The \emph{only} way not to get sensitive dependency is that the sequence
visits the neighbourhoods $A_k$ in order eventually, forever. 
For clearly if it would ever ``break the order'', as $|a_i - a_j| >> \epsilon$
the sensitive dependency would be immediate.  

So, this visiting the $A_k$s in order amounts to choosing one order among many
possible variations. Suppose, for a single iteration the probability 
that it would pick a certain order be $p$.
Clearly $0\le p \le 1$, and probability of the sequence always preserves the same order 
in $N$ iteration is $p^N$.
Obviously then, over infinite iteration it is an almost never event under the assumption
that $p \ne 1$. We argue that $p \ne 1$ is a reasonable assumption to hold.
The probability of picking a specific number $x$ from a range $[a,b]$
is zero under Lebesgue measure.
As this $p$ can take any value between $[0,1]$ to get into no sensitive dependence 
we need to get a precise $p=1$, but from previous discussion that assignment of probability 
is an almost never event. 

Hence the other event ``sometimes the sequence breaks order'' would be the almost sure event.
And that would immediately imply sensitive dependence as discussed earlier.

Thus, almost all sequences in $\mathcal{N}$ are chaotic (as in definition \ref{chaos-s}).   
\end{proof}
Or rather, $\mathcal{N}$ is almost entirely chaotic.
Which results in the next theorem.

\begin{theorem}\label{chaos-fpi}
\textbf{Function iteration in $X_Q$ is Almost Surely Chaotic.}

Almost all orbits of the function iteration of the type of equation \eqref{gen-cc-fpi}
are chaotic.
\end{theorem}
\begin{proof}
Using theorem \ref{fpim} and theorem \ref{ncc} the result immediately follows. 
\end{proof}

Thus, we establish that the computation like function iterations are 
almost certainly chaotic.

Now, in the light of computation let's analyse the same.
Suppose a computation does not terminate.
Then we are either in a loop, or not in a loop.
If we are in a loop, we can terminate trivially using another debugger machine
which maintains the state of the debugee in a separate table.

If we can never terminate, then, as we are not in a loop, 
we are either converging to a point in $\mathbb{R}$,
or we are not converging to any single point: we are having multiple accumulation points.
This brings the next theorem.

\begin{theorem}\label{chaos-c}
\textbf{Universal Computation is Almost Surely Chaotic.}

The sequence $c = <c_n>$ generated by rationalisation (definition \ref{rat}) 
of the tape of a Turing machine is almost surely chaotic in nature. 
\end{theorem}

\begin{proof}
We note that a terminating computation has a finite trajectory resulting in a finite sequence, 
so clearly they comprise of null set. A non terminating computation 
which yields finally a loop can be halted by maintaining the previous states in a separate table.
So, a computation ending in a loop is also finite in effect, and they also can be taken 
as null set because of finiteness.

Due to the nature of rationalisation (definition \ref{rat})
the only way $c$ can converge if the only symbols which gets changed (or written) are 
on the extreme right  of the tape eventually. 
Assigning a probability to the event of changing (or writing) one symbol at 
the extreme right of the tape as $0 \le p_R \le 1$,
we notice that $p_R = 1$ is required to meet the criterion of the final probability 
$\lim_{n \to \infty }p_R^n \ne 0$. 
This is not possible under the reasonable assumption that the $p_R$ is
taken at random from $[0,1]$ and the probability measure used is Lebesgue which is 0
for an isolated point. Thus, in the case of computation almost surely $c \not \in \mathcal{C}$. 

Then, almost surely $c \in \mathcal{N}$. Now, as $\mathcal{N}$ is almost entirely chaotic
(by theorem \ref{ncc}), the same argument also is applicable here, 
and hence $c$ must be almost surely chaotic.   

\end{proof}

We notice that this is similar to the problem of picking an arbitrary number from $[0,1]$
and asking what sort of number it would be? It clearly can not be 
a rational number, hence it must be irrational. But irrational numbers are almost entirely 
transcendental, so the number picked would be almost surely transcendental. 

\end{section}

\begin{section}{Learning Functions and Sets }\label{rgf}

In this section we do the opposite of the generation of sequences:
we try to learn from the sequence what sort of function would have generated it.
Or, in fact we are trying to learn the indicator function $1_A$ which was used to generate 
the set $A$ from which the numbers are being drawn.

What we want is to create a sequence of computable indicator functions which converges 
to the underlying indicator function.
This brings to the notion of computable functional, whose iteration would generate 
the sequence of the computable functions. But then, we need a way to map a function into $X_Q$.
Which requires the following definitions.

\begin{definition}\label{rcf}
\textbf{Rational Mapping of a Computable Function.}

Given an universal turing machine $\mathbb{U}$ (definition \ref{UTM}) any computable function can be encoded using the symbols
in the tape of the machine. The rationalisation (definition \ref{rat}) of the tape $\rho(T_C) $ then, 
serves as the rationalisation of the computable function.
$$
\rho(C) = \rho(T_C, \mathbb{U} ) = c \; ; \; C \in \mathbb{C}
$$
where $\mathbb{C}$ is the set of computable functions.
This makes $c \in X_Q$.  
\end{definition}

\begin{definition}\label{cf}
\textbf{The Computable Functional.}

A computable function $\mathcal{L} : X_Q \to X_Q $ is called a functional iff
the image of the function can be interpreted as an encoding of a computable function. 
That is :
$$
\forall x \; ; \; \exists C \in \mathbb{C} \; ; \; \rho(C) = \mathcal{L}(x) 
$$
\end{definition}

\begin{definition}\label{li}
\textbf{The Learning Iteration.}

Given a computable functional $\mathcal{L} : X_Q \to X_Q$, and input $T_S$
(training set) the learning iteration is defined as :
$$
c_{n+1} = \mathcal{L}(c_n)
$$    
with $c_0 = \rho(T_S)$.
Here $c_n$ signifies the rationalisation (definition \ref{rcf}) of the $C_n \in \mathbb{C}$. 
\end{definition}

Now, the problem begins with the insight that sequences from $X_Q$ is almost surely chaotic.

\begin{theorem}\label{lic}
\textbf{Learning Iteration is Almost Surely Chaotic.}

Learning iteration, as defined as definition \eqref{li} is almost surely chaotic.
\end{theorem}
\begin{proof}
We note : $\mathcal{L} : X_Q \to X_Q $, from (theorem \ref{chaos-c}) , 
that Universal computation is `almost surely' chaotic, the result is immediate.  
\end{proof}

But that is not all. Given that the function sequence is converging to a
limiting function, the problem is far from over.
That brings the next theorem.

\begin{theorem}\label{licr}
\textbf{Almost Sure Non-Computability for converging function.}

The iteration of definition \eqref{li}, when converges, almost surely converges to an un-computable function.
\end{theorem}

\begin{proof}
This is trivial from Real analysis.
We know that the sequence $c_n \in X_Q$, so to complete the space (definition \ref{cncms} ) 
we need to have  the embedding space  $\mathbb{R}$. 
Clearly then, almost all limit points would be irrational 
( actually transcendental ), with a cardinality equal to continuum $|\{ lim<c_n> \}| = \aleph_1$
which is a well known result from the real analysis \cite{hlr}.
Irrational numbers take infinitely many symbols to encode, 
therefore, the function limit $lim<c_n>$ can not be encoded in a Turing machine tape, 
and hence, is obviously non computable.
 
This is precisely what we wanted to show.
\end{proof}
 
\end{section}

\begin{section}{Fractals in the Middle}\label{fm}

Chaos always gets accompanied by a fractal.
We showcase that the computation is not any different.
However, the occurrence of fractal comes from most unexpected of the area.
We ask the question : When a sequence of computable indicator function converges, 
what is the nature of the limiting set indicated by the limit of the sequence? 

As an example of computable indicator functions converging, we take Cantor's construction.
The limiting set is the cantor set as defined by:
$$
 C = \bigcap_{m=1}^\infty \bigcap_{k=0}^{3^{m-1}-1} \left(\left[0,\frac{3k+1}{3^m}\right] \cup \left[\frac{3k+2}{3^m},1\right]\right)
$$

Cantor set is a well known example of a fractal set \cite{hjss}\cite{cfnfs}. 

Therefore, can we conjecture that 
a sequence of Converging computable indicator functions in the limit actually accepts a fractal subset
of $X_Q$ ? As we would find next, it exactly does that.

But before we formally show that we need some definitions, so that we can generalise the construction 
of the Cantor Set into the general setting of any computable indicator function family 
recursively getting used. The specific construct is a binary decision tree. 

\begin{definition}\label{bdt}
\textbf{Binary Decision Tree.}

A binary tree with each nodes having two decider Turing Machine (definition \ref{decider}) 
from a finite set of decider machines $(D_f,D_d ) \in \mathbb{D}$ such that:
$$
D_f : X_Q \to X_Q  \text { and } D_d : X_Q \to \{0,1,l,r\}  
$$
Given $x \in Z$ with $Z \subseteq X_Q$ is the input to the current node, let : $y = D_f(x) $ and $d = D_d(y) $.
If $d=0$ the input is rejected and the system is halted.
If $d=1$ the input is accepted and the system is halted.
If $d = l$ then $y$ is passed as input to the left child ; if no left child exists, reject input x and halt.
If $d = r$ then $y$ is passed as input to the right child ; if no right child exists, reject input x and halt.
\end{definition}

This structure (definition \ref{bdt}) partitions the input space into finite number of equivalent partitions such that:
$$
X_Q = \bigcup_i P_i
$$
with either $ P_i  \subset A$ or $ P_i  \subset X_Q \setminus A$, but can not be both. 
We also note that if $P_i \subset A$ then $P_{i+1} \subset X \setminus A$. 
Clearly then $A$ is a disjoint union of sets:
$$
A = \bigcup a_i \; ; \; a_i \cap a_j = \varnothing \; ; \; i \ne j 
$$

To demonstrate that this structure impose self-similarity, 
we need a definition of self-similarity.

\begin{definition}\label{ss}
\textbf{Self Similarity.}

Let there be a topological space $X$ (definition \ref{top-space}) and a set of non-surjective homeomorphic functions $\{ f_s\}$
(definition \ref{hom}) indexed by a finite index set $S = \{ s \}$ with :
$$
X = \bigcup_{s \in S} f_s(X) \; .
$$
If $X\subset Y$, we call $X$ self-similar if it is the only non-empty subset of $Y$ 
such that the equation above holds for.
Then $\mathcal{S} = ( X , S , \{ f_s \} )$ is called a self similar structure.
\end{definition}

\begin{theorem}\label{idtss}
\textbf{Infinite Binary Decision Trees have Self Similar Structure.}

A binary decision tree as defined in \eqref{bdt} is called infinite
if it has countably infinite nodes. 
Accept set of such a tree exhibits self similar structure ( definition \ref{ss} ).
\end{theorem}
\begin{proof}
We note that every node in the tree first does a transform 
of the input space $X_Q$ using $D_f$ into $Y$ ; both homeomorphic to $\mathbb{Q}$. 
After that, this space $Y$ is partitioned using 
$D_d$ which are finite in number as discussed before.
Then this individual partitions are homeomorphic to $\mathbb{Q}$.
Suppose $n(D_d)$ defines the number of partitions. We can then say there are $n(D_d)$
numbers of  homeomorphic (on $\mathbb{Q}$) non-surjective function available for each $(D_d, D_f)$.
This is due to lemma \eqref{p-hom}.

The set of such functions is finite because $\forall n \; ; \; n(D_d)$ and $\mathbb{D}$ is finite.
Suppose then, the set of such composition is termed as $\mathbb{F}$.
It is then trivial that when $X = X_Q$ (at the root) then : 
$$
X_Q = \bigcup_{F \in \mathbb{F} } F( X_Q )
$$
Therefore, infinite binary decision tree have self similar structure. 
\end{proof}
In fact it is well known that \emph{Cantor Set} \cite{hjss}\cite{gc}\cite{cfnfs} is 
specific example of this sort of structure ( a dyadic Tree ). 

\begin{theorem}\label{cliss}
\textbf{Convergent Learning Function `almost surely' have a Self Similar Structure.}

If a learning iteration (definition \ref{li}) converging to a function $f_A$,
then accept set of $f_A : A$ almost surely would have a self similar structure. 
\end{theorem}
\begin{proof}
Almost surely the function is un-computable (theorem \ref{licr}). 
We note that the due to the fixed size of the computable functional, 
the number of deciders the learning system must stay finite. 
And therefore, to be convergent the structure is to be extended to infinity :
that is an infinite binary decision tree (definition \ref{bdt}).
Now, using theorem \eqref{idtss}, the result is immediate. 
\end{proof}

\end{section}

\begin{section}{Summary}\label{sm}

To tackle fixed point iterations like universal computation
the term Chaos has been redefined for functions whose domain and range are in $X_Q$ and 
are in general continuous nowhere. The definition showcases the hallmark property of chaos :
sensitivity to initial condition (butterfly effect) very clearly. 

As a working definition: \emph{``a bounded non converging non repeating sequence
almost surely is exhibiting chaos''} has been suggested. 
This simplified definition is most useful where the underlying function can not be studied analytically, 
just as it so happened with universal computation.

The learning functional describes fully autonomous machine learning.
It opens up for the debate that how sentient life forms really learn. 
Certainly the set accepted in the limit would be a fractal, 
which surely would be different for different learners. 
This raises serious issues on the computability 
aspects of autonomous machine learning techniques like deep learning.

It is understood that these findings are highly counter intuitive.
However, as it has been shown before that the Dynamical Systems capable of exhibiting chaos can simulate Universal Computation, these findings should not come as surprise. 
This is a closure of the study in this regard, as well as the new beginning from applicability perspective.

\end{section}

\appendix
\begin{section}{Definitions Used}\label{ap_1}
\begin{definition}\label{fp}
\textbf{Fixed Point of a function. }

For a function $f:X \to X$ , $x^*$ is said to be a fixed point, iff $f(x^*) = x^*$ .
\end{definition}

\begin{definition}\label{mp}
\textbf{Metric Space.}

A metric space is an ordered pair $(M,d)$ where $M$ is a set and $d$ is a metric on $M$ , i.e., a function:-

$$
d : M \times M \to \mathbb{R}
$$

such that for $x,y,z \in M$ , the following holds:-
\begin{enumerate}
\item { $d(x,y) \ge 0 $ }
\item { $ d(x,y) = 0 $ iff $x=y$ . }
\item { $d(x,y) = d(y,x) $ }
\item { $d(x,z) \le  d(x,y) + d (y,z)$ }
\end{enumerate}

The function `$d$' is also called ``distance function'' or simply ``distance''.
\end{definition}

\begin{definition}\label{cs}
\textbf{Cauchy Sequence in a Metric Space $(M,d)$ .}

Given a Metric space $(M,d)$ , the sequence $x_1,x_2,x_3,...$ of real numbers is called `Cauchy Sequence', if for every positive real number $\epsilon$ , 
there is a positive integer $N$ such that for all natural numbers $m,n > N$ the following holds:-
$$
d (x_m , x_n ) < \epsilon .
$$

\end{definition}
Roughly speaking, the terms of the sequence are getting closer and closer together in a way 
that suggests that the sequence ought to have a limit $x^*  \in M$ . Nonetheless, such a limit does not always exist within $M$ .

Note that by the term: \emph {sequence} we are implicitly assuming \emph {infinite sequence} , unless otherwise specified.

\begin{definition}\label{cncms}
\textbf{Complete Metric Space.}

A metric space $(M,d)$ is called complete (or Cauchy) 
iff every Cauchy sequence (definition \ref{cs}) of points in $(M,d)$ has a limit , that is also in $M$ .
\end{definition} 

As an example of not-complete metric space take $\mathbb{Q}$ , the set of rational numbers. 
Consider for instance the sequence defined by $x_1 = 1$ and function $d$ is defined by standard
difference between $d(x,y) = |x-y|$ , then :-
$$
x_{n+1} = \frac{1}{2} \left ( x_n + \frac{2}{x_n} \right ) 
$$

This is a Cauchy sequence of rational numbers,
but it does not converge towards any rational limit, but to
$\sqrt{2}$ , but then $\sqrt{2} \not \in \mathbb{Q}$ . 

The closed interval $[0,1]$ is a Complete Metric space which is homemorphic (definition \ref{hom}) to $\mathbb{R}$.

\begin{definition}\label{orbit}  \cite{ap} \cite{mbgs}
\textbf{Orbit.}

Let $f:X \to X$ be a function. 
The sequence $ \mathcal{O} = \{x_0, x_1,x_2,x_3,...\}$ where
$$
x_{n+1} = f(x_n) \; ; \; x_n \in X \; ; \; n \ge 0 
$$
is called an orbit (more precisely `forward orbit') of $f$. 

$f$ is said to have a `closed' or `periodic' orbit $ \mathcal{O}$ if $| \mathcal{O}| \ne \infty$ .
\end{definition}

\begin{definition}\label{top-space}
\textbf{Topological Space.}

Let the empty set be written as : $\emptyset$. Let $2^X$ denotes the power set, i.e. the set of all subsets of $X$.
A topological space is a set $X$ together with $\tau \subseteq 2^X$ satisfying the following axioms:-
\begin{enumerate}
\item{ $\emptyset \in \tau$ and $X \in \tau$ ,}
\item{ $\tau$ is closed under arbitrary union, }
\item{ $\tau$ is closed under finite intersection. }
\end{enumerate}
The set $\tau$ is called a topology of $X$.
\end{definition} 

\begin{definition}\label{hom}
\textbf{Homeomorphism.}

A function $f: X \to Y$ between two topological spaces $(X, T_X)$ and $(Y, T_Y)$ is called a homeomorphism if it has the following properties:
\begin{enumerate}
\item{ f is a bijection (one-to-one and onto),}
\item{f is continuous,}
\item{the inverse function $f^{-1}$ is continuous (f is an open mapping).}
\end{enumerate}

\end{definition}

\begin{lemma}\label{p-hom}
\textbf{Existence of Homeomorphic functions on Partitions.}

Let $\mathbb{P} = \{ P_i \}$ be a partition of $X$ homeomorphic to $\mathbb{Q}$, such that:
$$
\bigcup_i^n P_i = X \; ; \; \forall i \ne j \; P_i \cap P_j = \varnothing  
$$     
then, there exists $n = |\mathbb{P}|$ homeomorphic functions from $X \to P_i$.
\end{lemma}

\begin{definition}\label{bounded-seq}
\textbf{ Bounded Sequence.}

A sequence $<x_n>$ is called a bounded sequence iff :-
$$
\forall n \; \; l \le x_n \le u \; ; \;  -\infty < l  \le  u < \infty  
$$

The number `l' is called the lower bound of the sequence and 
`u' is called the upper bound of the sequence.
\end{definition}

\begin{lemma}\label{bcss}
\textbf{Bolzano-Weierstrass.}

Every bounded sequence has a convergent (Cauchy) subsequence.
\end{lemma}

It is to be noted that a bounded sequence may have many convergent subsequences (for example, a sequence consisting of a counting of the rationals has subsequences converging to every real number) or rather few (for example a convergent sequence has all its subsequences having the same limit).

\begin{definition}\label{TM}
\textbf{Turing Machine.}

A ``Turing Machine'' is a 7-tuple ($Q,\Sigma,\Gamma,\delta,q_0,q_a,q_r$), where:-
\begin{enumerate}
\item{ $Q$ is the set of states. }
\item{ $\Sigma$ is the set of input alphabets not containing the blank symbol $\beta$. }
\item{ $\Gamma$ is the tape alphabet , where $\beta \in \Gamma$ and $\Sigma \subseteq \Gamma$. }
\item{ $\delta : Q \times \Gamma \to Q \times \Gamma \times \{L , R \} $ is the transition function. }
\item{ $q_0 \in Q$ is start state.}
\item{ $q_a \in Q$ is the accept state.}
\item{ $q_r \in Q$ is the reject state.}
\end{enumerate}
\end{definition} 
According to standard notion $q_a \ne q_r$ , but we omit this requirement here, 
as we are not going to distinguish between two different types of halting (`accept and halt' vs `reject and halt') of Turing Machines.

A Turing Machine `$M$' (definition \ref{TM}) computes as follows.

Initially `$M$' receives the input $w=w_1w_2w_3...w_n \in \Sigma^* $ 
on the leftmost `$n$' squares on the tape, and the rest of the tape is filled up with blank symbol `$\beta$'.
The head starts on the leftmost square on the tape. 
As the input alphabet `$\Sigma$' does not contain the blank symbol `$\beta$', 
the first `$\beta$' marks end of input.

Once `$M$' starts, the computation proceeds wording to the rules of `$\delta$'.
However, if the head is already at the leftmost position, then, 
even if the `$\delta$' rule says move `$L$' , the head stays there.

The computation continues until the current state of the Turing Machine is either $q_a$ , or $q_r$ .
In lieu of that, the machine will continue running forever.

\begin{definition}\label{decider}
\textbf{Decider Turing Machine.}

A Turing Machine, which is guaranteed to halt on any input (i.e. reach one of the states \{$q_a,q_r$\} ) is called a decider.
\end{definition}

\begin{definition}\label{undecidable}
\textbf{Undecidable Problem.}

If for a given problem, it is impossible to construct a  decider (definition \ref{decider}) Turing Machine, 
then the problem is called undecidable problem.
\end{definition}

\begin{definition}\label{UTM}
\textbf{Universal Turing Machine.}

An `UTM' or `Universal Turing Machine' is a Turing Machine (definition \ref{TM}) such that it can simulate an 
arbitrary Turing machine on arbitrary input.
\end{definition}

\begin{definition}\label{CTT}
\textbf{Church Turing Thesis.}

Every effective computation can be carried out by a Turing machine (definition \ref{TM}), 
and hence by an Universal Turing Machine(definition \ref{UTM}).
\end{definition}

\begin{definition}\label{god}
\textbf{G\"{o}delization (G\"{o}del).}

Any string from an alphabet set $\Gamma$ can be represented as an integer in base `$b$' with $b = |\Gamma|$. 
To achieve this, create a one-one and onto G\"{o}del map $g : \Sigma \to D_b$ , where,
$$
D_b = \{ 0 , 1, 2, ... ,b-1 \} .
$$ 
G\"{o}delization or $\mathbb{G} : \Sigma^+ \to \mathbb{Z_+}$ then, is defined as follows:
 
A string of the form $w = w_{n-1}w_{n-2}...w_1w_0$ , with  $w_i \in \Gamma$ ,
can be mapped to an integer $I_w = \mathbb{G}(w)$ as follows \cite{gm}:
$$
I_w = \mathbb{G}(w) = \sum\limits_{k=0}^{n-1} g(w_k) b^{k}
$$

\end{definition}
The common decimal system is a typical example of G\"{o}delization of symbols from $\{ 0,1,..,8,9\}$.
The binary system represents G\"{o}delization of symbols from  $\{ 0,1\} $.
As a far fetched example, any string from  the whole english alphabet, can be written as a base 26 integers!
 
\begin{definition}\label{rat-s}
\textbf{Rational Bounded Space.}

The set of rational numbers within $[0,1]$, i.e.
$$
X_Q = \mathbb{Q}\cap[0,1]
$$
is called rational bounded space, and would be designated by $X_Q$.
\end{definition}

\begin{definition}\label{rat}
\textbf{Rationalization. }

Any string `$w$' of length `n' ($|w|=n$) , created from an alphabet set $\Gamma$,
can be represented as a rational number $x \in \mathbb{Q}$.
We define the rationalization, $\rho$ , in terms of G\"{o}delization (definition \ref{god}) as follows \cite{gm}:
$$
x = \rho(w) = \frac{\mathbb{G}(w)}{ b^n } = \mathbb{G}(w) b^{-n} = 0.w_{n-1}w_{n-2}...w_0
$$
By definition, $x \in X_Q$ (definition \ref{rat-s}).
\end{definition}

\end{section}


\begin{thebibliography}{9}


\bibitem{ap}
 {\sc Arrowsmith, D.K. and Place, C.M.}, 
 \emph{An Introduction to Dynamical Systems}. 
 Cambridge University Press.


\bibitem{mbgs}
 {\sc Brin, Michael and Stuck, Garrett.}, 
 \emph{Introduction to Dynamical Systems}. 
 Cambridge University Press.

\bibitem{cfnfs}
  {\sc Heinz-Otto Peitgen,	Hartmut J\"{u}rgens , Dietmar Saupe},
  \emph{Chaos and Fractals New Frontiers of Science}.
  Springer.


\bibitem{turing}
 {\sc Turing, A. M.}, 
 \emph{ On Computable Numbers, with an Application to the Entscheidungsproblem }.
 Proceedings of the London Mathematical Society, 2(42): 230-265


\bibitem{wf}
  {\sc Feller, W. },
  \emph{An Introduction to Probability Theory and Its Applications}.
   Vol. 2, 3rd ed. New York: Wiley,  1968.

\bibitem{wd}
  {\sc Williams, David. },
  \emph{Probability with Martingales}.
   Cambridge University Press, 1991.

\bibitem{mlpv}
 {\sc by Ming Li, Paul M.B. Vitanyi}, 
 \emph {An Introduction to Kolmogorov Complexity and Its Applications (Texts in Computer Science)}. 
  Springer (21 Nov 2008).


\bibitem{hlr}
 {\sc Royden, H. L. }, 
 \emph {Real Analysis }. 
 Prentice Hall of India Private Limited, 2003.


\bibitem{hjss}  
 {\sc Smith,Henry J.S. }  
 \emph{  On the integration of discontinuous functions}. 
 Proceedings of the London Mathematical Society, Series 1, vol. 6, pages 140-153, 1874.

\bibitem{gc}
{\sc Cantor,George. }    
\emph{ \"Uber unendliche, lineare Punktmannigfaltigkeiten V [On infinite, linear point-manifolds (sets)]}.
 Mathematische Annalen, vol. 21, pages 545-591,1883.  


\end{thebibliography}
\end{document}